\documentclass[11pt]{article}
\usepackage{fullpage}
\usepackage[utf8]{inputenc}
\usepackage{amsmath, amsthm, amssymb, mathtools,xcolor}

\usepackage{hyperref}
\hypersetup{
    colorlinks=true,
    allcolors=blue
}

\newcommand{\ket}[1]{|#1\rangle}
\newcommand{\bra}[1]{\langle#1|}
\newcommand{\ketbra}[2]{|#1\rangle\langle#2|}

\newcommand{\norm}[1]{\left\lVert#1\right\rVert}

\newtheorem{theorem}{Theorem}[section]

\newtheorem{lemma}[theorem]{Lemma}
\newtheorem{claim}[theorem]{Claim}

\def\01{\{0,1\}}

\newcommand{\eps}{\varepsilon}
\renewcommand{\eps}{\varepsilon}

\newcommand{\mathify}[1]{\ifmmode{#1}\else\mbox{$#1$}\fi}

\def\ri{{\rm i}}

\newcommand{\Rolando}[1]{{\color{red} [{\bf Rolando:} #1]}}
\newcommand{\Ronald}[1]{{\color{blue} [{\bf Ronald:} #1]}}

\title{Optimal Lower Bound for\\ Ground-State Energy Estimation 
with a Guiding State}
\author{Rolando D. Somma\thanks{Google Quantum AI, Venice, CA 90291, United States.}
\and 
Ronald de Wolf\thanks{Google Quantum AI, Venice, CA 90291, United States. Also QuSoft, CWI and University of Amsterdam, the Netherlands. Partially supported by the Dutch Research Council (NWO) through Gravitation-grant Quantum Software Consortium, 024.003.037.}
}
\date{}

\begin{document}

\maketitle

\begin{abstract}
Suppose we can apply the unitary $U=e^{\ri H}$ for some Hamiltonian $H$, and are given access to a unitary that prepares a guiding state promised to have  overlap at least $\gamma >0$ with the ground space of $H$. Our goal is to estimate the ground-state energy of $H$ within additive error $\delta > 0$ and success probability at least $1-\eps$, $\eps>0$.
How many applications of $U$ and its inverse $U^{-1}$ are necessary and sufficient?
This quantity corresponds to the total Hamiltonian-simulation time needed.
An upper bound $O(\log(1/\eps)\log(1/\gamma)/\gamma\delta)$ was known, and was improved to $O(\log(1/\eps)/\gamma\delta)$ very recently~\cite{JW:optQPE}. A matching lower bound was known whenever one of the three parameters $\delta,\gamma,\eps$ was held constant~\cite{Mande2026tightboundsquantum}.
In this paper we prove the joint lower bound $\Omega(\log(1/\eps)/\gamma\delta)$ with the tight $\eps$-dependence provided the dimension of $H$ is at least $\log(1/\eps)/\gamma^2$. Furthermore, 
we show that this same lower bound (with slightly larger dimension) holds for both the special case in which the ground state is guaranteed to be unique and $H$ has a gap of $\delta$ between its first and second eigenvalue; and
for ground-state preparation, where $\delta$ denotes the spectral gap and $\eps$ now is the approximation error. The lower bounds also apply when the Hamiltonian can be accessed via its block-encoding, and when fractional powers of $U$ are allowed, as in continuous-time Hamiltonian simulation. Lastly, improved upper bounds are known when $H$ is nonnegative
and presented as a sum of squares; and our results imply the lower bound $\Omega(\log(1/\eps)/\gamma\sqrt{\delta})$ for this case.
\end{abstract}



\section{Introduction}

\subsection{Ground-state energy estimation with guiding states}

One of the core problems in quantum chemistry is the following: given a classical description of some Hamiltonian $H$ (e.g., an electronic structure Hamiltonian), estimate its \emph{ground-state energy}, which is its smallest eigenvalue $\lambda_{\min}$. If $H$ is normalized such that its eigenvalues are all in $[0,2\pi-2\delta)$ (the ``$-2\delta$'' is to avoid the issue that 0 and $2\pi$ are the same around the unit circle) and we define the unitary $U=e^{\ri H}$ (which has the same eigenvectors as $H$, with eigenvalue $\lambda$ of $H$ becoming eigenvalue $e^{\ri \lambda}$ for $U$), then estimating the ground-state energy of $H$ is equivalent to estimating the smallest eigen\emph{phase} of $U$. 

If we are additionally given a \emph{ground state} (i.e., an eigenstate of~$H$ corresponding to $\lambda_{\min}$), then the well-known phase estimation algorithm~\cite{Kit95} is tailor-made to estimate the ground-state energy: it can approximate $\lambda_{\min}$ up to $\pm\delta$ using $O(1/\delta)$ applications (or ``queries'') of $U$, with small constant error probability. However, it is typically hard to prepare the ground state of $H$, or even something close to it. What \emph{can} sometimes be done in situations where we (or our chemistry friends) have some intuition about roughly what the ground state should look like, is a description and/or relatively cheap preparation of some quantum ``guiding state'' (also sometimes known as ``Ansatz'') that has some non-negligible overlap with the ground state.
In the complexity-theoretic context, this problem of ground-state energy estimation for a local Hamiltonian given a guiding state is known as the 
``guided local Hamiltonian problem''~\cite{Ric07,GL23,Wai25}. There is also the ``guidable'' variant, where a guiding state is promised to exist in some class, but is not actually given~\cite{Bra15,WFC24,CFG+23}. These problems have connections with deep complexity questions such as the PCP conjecture

In this paper we study a constrained version of the guided Hamiltonian problem, where we cannot directly access the structure of $H$ (e.g., make use of its locality of various kinds) but can only perform Hamiltonian simulation according to~$H$ in the form of the unitary $U=e^{\ri H}$.
Similarly we cannot directly access the structure of the guiding state, all we have is a unitary to prepare it.

We focus first on the minimal number of times we need to apply~$U$ and its inverse $U^{-1}$ to estimate $\lambda_{\min}$ (or $\lambda_{\max}$, which can be seen to be an equally hard problem basically by simply replacing $H$ by $2\pi I-H$), in terms of the allowed approximation error~$\delta$, the overlap parameter~$\gamma$, and the allowed error probability~$\eps$:
\begin{quote}
Let $H$ be a Hamiltonian, with eigenvalues normalized to be in $[0,2\pi-2\delta)$, whose lowest eigenvalue we want to approximate. 
We can apply the unitary $U=e^{\ri H}$, its inverse $U^{-1}$, as well as a unitary~$A$ (and $A^{-1}$) that can generate a guiding state $A\ket{0}$ that is promised to have  overlap $\geq\gamma$ 
with the ground state of $H$
(or ground space, if it is degenerate).
How many applications of $U$ and/or $U^{-1}$ are necessary and sufficient to estimate the ground-state energy of $H$ within $\pm\delta$ with success probability $\geq 1-\eps$?
\end{quote}
Ground-state energy estimation is a core problem in quantum chemistry and in quantum computing, and the number of applications of $U$ tends to be the dominant cost in algorithms for this problem, so determining the optimal bounds for this is an important question.
Note that we are only counting the number of applications of $U$ and $U^{-1}$ here, disregarding the gate-complexity and the number of applications of $A$ and $A^{-1}$.
We are also disregarding the number of auxiliary qubits. These ``disregardings'' only strengthen our lower bounds on the number of applications of $U$ and $U^{-1}$, since those lower bounds remain valid even if one allows an unlimited number of applications of $A$ and $A^{-1}$ (allowing for instance for arbitrarily precise tomography of the guiding state $A\ket{0}$), and unlimited number of gates and auxiliary qubits.

One can think of $U=e^{\ri H}$ (or $U^{-1}=e^{-\ri H}$) as performing Hamiltonian simulation with respect to $-H$ (or $H$) for one unit of time. We could also allow unitaries of the form $e^{\ri H\tau}$ for non-integer~$\tau$, allow an algorithm to invoke $e^{\ri H\tau_1},\ldots,e^{\ri H\tau_s}$ for times $\tau_1,\ldots,\tau_s$ of its choice, and count the total cost of the algorithm as $\sum_{i=1}^s|\tau_i|$~\cite{Cleve2009Efficient}.
Our results easily carry over to this setting.\footnote{For example, in the hard families for the lower bounds one can redefine $H'=H/1000$ and $U'=e^{\ri H'}$, then we allow 1/1000-th fractional queries to the original $U$.}

An upper bound $O(\log(1/\eps)\log(1/\gamma)/\gamma\delta)$ on the number of applications of $U$ and $U^{-1}$ has been known for some years~\cite{linlin&tong:groundstateprep,Mande2026tightboundsquantum}.
In particular, \cite[Lemma 4.5]{Mande2026tightboundsquantum} obtains an $O(\log(1/\gamma)/\gamma\delta)$-algorithm that has small constant success probability by, roughly speaking,  putting a quantum minimum-finding algorithm on top of phase estimation with precision~$\delta$ that starts from the guiding state.
The $\log(1/\gamma)$-factor is there to reduce the error probability of the ``inner'' algorithm (phase estimation)  to $\ll\gamma^2$ in order to make the composition work.
The error probability can then be reduced to small~$\eps$ by running this algorithm $O(\log(1/\eps))$ times and essentially taking the median of the outputs (being careful about the ``circularity'' of $[0,2\pi)$), giving an $O(\log(1/\eps)\log(1/\gamma)/
\gamma\delta)$ upper bound for our problem.

Very recently, simultaneously to and independently of this paper, Jeffery and Witteveen~\cite{JW:optQPE} used the method of ``transducers'' to implement the composition in a better way that avoids the extra factor $\log(1/\gamma)$, giving upper bound $O(1/\gamma\delta)$ for small constant error probability.
Again, the error probability can be reduced to $\eps$ quite straightforwardly by running the algorithm $O(\log(1/\eps))$ times, giving a tight $O(\log(1/\eps)/
\gamma\delta)$ upper bound.%
\footnote{For the two specific hard families of instances in our paper, it is fairly easy to see that the $O(\log(1/\eps)/\gamma \delta)$ upper bound holds: because we are dealing with only two possible known eigenphases, we can do exact phase estimation to distinguish those two values, and then use amplitude amplification on top of that to search for a nonzero eigenphase with error probability $\leq\eps$. Proving the upper bound in general is quite involved, however.}

Matching \emph{lower} bounds are known whenever one of the three parameters $\delta,\gamma,\eps$ is $\Omega(1)$:
\begin{itemize}
    \item $\Omega(1/\gamma\delta)$ for constant~$\eps$ is Lemma~3.6 of~\cite{Mande2026tightboundsquantum}
    \item $\Omega(\log(1/\eps)/\delta)$ for constant~$\gamma$ is Theorem~1.3 of~\cite{Mande2026tightboundsquantum}
    \item    $\Omega(\log(1/\eps)/\gamma)$ for constant~$\delta$ follows from Theorem~4 of~\cite{BCWZ99} with $t=\gamma^2 N$, and $N$-dimensional unitary $U=\mbox{diag}(\{(-1)^{x_j}\}_{j\in[N]})$, 
    which corresponds to Hamiltonian $H=\mbox{diag}(\{\pi x_j\}_{j\in[N]})$, where $[N]=\{1,\ldots,N\}$. This $U$ is a ``phase-query'' to the $N$-bit string $x=x_1 \ldots x_N$, and determining whether $\lambda_{\max}$ is 0 or $\pi$ computes the OR function on~$x$.
\end{itemize}
The linear dependence in $\log(1/\eps)$ of the second and third bullets is disappointing, because for instance for quantum search one can reduce the constant error probability of Grover's $O(\sqrt{N})$-query algorithm to subconstant~$\eps$ at the expense of only a factor~$\sqrt{\log(1/\eps)}$~\cite{BCWZ99}.
Determining the optimal dependence on subconstant~$\eps$ is important, for instance when using ground-state estimation as a subroutine within a larger algorithm.
For subroutines that have one output with large probability one can use transducer-techniques~\cite{belovs2024taming} to avoid having to pay for reducing the error probability of the subroutine, but for subroutines with multiple outputs---such as ground-state estimation---no such efficient general method is known (though one can think of \cite{JW:optQPE} as doing something similar in this special case).

Proving a lower bound that is optimal in all three parameters simultaneously was left open in~\cite{Mande2026tightboundsquantum}, and is one of the main results of this paper: 
$\Omega(\log(1/\eps)/\gamma\delta)$ applications of $U$ and $U^{-1}$ are necessary 
(and, thanks to~\cite{JW:optQPE}, also sufficient).
The result in the third bullet is very suggestive of an approach: replace the $\pi x_j$ in $H$ by $3\delta x_j$, then approximating $\lambda_{\max}$ up to $\pm\delta$ still computes the OR on $x$. Since computing the bit $x_j$ from the modified unitary $U=e^{\ri H}$ now takes $\Theta(1/\delta)$ queries instead of~one, it is very plausible that this modification multiplies the lower bound of the third bullet by $1/\delta$, giving the conjectured $\Omega(\log(1/\eps)/\gamma\delta)$ bound. This is not a proof, however, and using for instance composition theorems for the negative-weights adversary lower bound~\cite{hls:madv,reichardt:tight,lmrss:stateconv,belovs:variations,BelovsLee2020Composition} 
  does not work well when we want to prove bounds for subconstant error probability~$\eps$.

\subsection{Our results}

\paragraph{First lower bound.}
In Section~\ref{sec:lowerboundnonunique} we prove (with help from ChatGPT, see our AI statement at the end of the paper) the lower bound $\Omega(\log(1/\eps)/\gamma\delta)$ if the dimension is $>\log(1/\eps)/\gamma^2$.
The proof uses the polynomial method to analyze the hard family of instances suggested in the previous paragraph. 
It shows, roughly speaking, that solving the $\delta$-approximate case with $T$ queries gives a solution to the $O(1)$-approximate case with $T\delta$ queries. The latter can be lower bounded using the analytic tools (Coppersmith-Rivlin~\cite{coppersmith&rivlin:poly} plus the extremal properties of Chebyshev polynomials) that were used for the $N$-bit OR function in~\cite{BCWZ99}.

Conceptually one may think of our lower bound as adding a factor $1/\delta$ to the known bound $\Omega(\log(1/\eps)/\gamma)$ mentioned in the third bullet above. This approach actually works in general when one already has a polynomial-based lower bound for regular phase queries, and wants to derive a lower bound for fractional queries.

\paragraph{Lower bound for the special case where the ground state is unique.}
Our first lower bound answers the open question from~\cite{Mande2026tightboundsquantum} but is not fully satisfactory: the most interesting situation is where the ground state of~$H$ is unique, with some gap between the first two eigenvalues, and the guiding state has overlap at least $\gamma$ with that unique eigenstate.
In contrast, in our lower bound, the ground space has dimension $\geq\gamma^2 N$, and the guiding state $A\ket{0}=\frac{1}{\sqrt{N}}\sum_{j\in[N]}\ket{j}$ has overlap $\gamma$ spread uniformly throughout that space.

In Section~\ref{sec:lowerboundunique} we modify the proof to get the same lower bound for a hard family of instances with a unique ground state, and a gap of $3\delta$ between the first two eigenvalues of~$H$.
The downside is that we have to make the dimension larger by a $\log(1/\eps)$-factor now.

\paragraph{Hardness of finding an approximate ground state of $H$.}
Instead of just approximating the ground-state \emph{energy} to within $\delta$, one might also be interested in preparing a (mixed or pure) ground \emph{state} that is close to the ground state in trace distance. This is assuming $\gamma$ is much less than~1, of course, for otherwise the guiding state is already close to the ground state. In our first hard family (Section~\ref{sec:lowerboundnonunique}), the ground states are basis states $\ket{j}$, for the $\gamma^2 N$ indices $j$ such that $x_j=1$, while in our second hard family (Section~\ref{sec:lowerboundunique}) the unique ground state is $\sqrt{a}\ket{0}+\omega\sqrt{1-a}\ket{j}$ for $a\geq\gamma^2$, which (for small $a$) is essentially the unique $\ket{j}$.

Suppose we can prepare the unique ground state of $H$ 
within trace-distance error $\eps$,
where the spectral gap (i.e., the difference between $\lambda_{\min}$ and the next eigenvalue of $H$) is at least $\delta$. As before, 
we would like to achieve this state preparation assuming access to $U$ and to the unitary~$A$ that prepares the guiding state.
Our second result essentially proves that the same lower bound $\Omega(\log(1/\eps)/(\gamma \delta))$ applies to the number of applications of $U$ and $U^{-1}$ in this case, provided $\gamma$ is small (of course, when $\gamma=1$ the guiding state \emph{is} the ground state). This follows by a simple reduction from the unique ground-state energy estimation problem, as ground states can be used to estimate $\lambda_{\min}$ via for example quantum phase estimation.

\paragraph{Starting from a block-encoding of $H$.}
While our results are formulated
for the query complexity measured by the number of uses of $U$ and $U^{-1}$, we can also consider lower bounds in the related block-encoding model. In this case, the Hamiltonian is accessed via a unitary $V$ acting on a space of larger dimension $M \ge N$, such that its upper-left block contains the Hamiltonian: $\bra 0 V \ket 0 = H$. Here the eigenvalues of~$H$ are all in $[-1,1]$.
The block-encoding model has been proven useful for a variety of tasks, including phase estimation~\cite{babbush2018encoding} and ground-state preparation~\cite{linlin&tong:groundstateprep}, and combined with quantum singular-value transformation unifies many quantum algorithms~\cite{gilyenea:svtrans,Martyn2021GrandUnification}. 

Our lower bounds on the number of uses of $U$ (and $U^{-1}$) also extend to the number of uses of $V$ (and $V^{-1}$) in this block-encoding model
via another reduction, where the Hamiltonian now is $H=-\ri (U -U^{-1})/2$. The block-encoding unitary is simply $V =-\ri (\ketbra ++ \otimes U - \ketbra --\otimes U^{-1})$, 
where $\ket +=( \ket 0 + \ket 1)/\sqrt 2$ and $\ket -=(\ket 0 - \ket 1)/\sqrt 2$. This is like a controlled-$U$ operation conjugated by a single-qubit rotation, and satisfies the desired property $\bra 0 V \ket 0 = H$. If the eigenvalues of $U$ are $e^{\ri \lambda}$, then the eigenvalues of $H$ are $\lambda'=\sin(\lambda)$. An estimate of $\lambda'$
within error $\delta'$ provides an estimate of $\lambda$ within error $\delta \approx \delta'/\sqrt{1-(\lambda')^2}$;
assuming, for example, $\lambda \in [0,\pi/2)$ so that $|\lambda'|<1$, then estimating $\lambda'$ within error $O(\delta)$ gives $\lambda$ within error $\delta$, and the lower bound carries over to the estimation of $\lambda'_{\min}$. That is, since each block-encoding unitary $V$ uses $U$ and $U^{-1}$ once, the lower bound also applies to the number of calls to the block-encoding~$V$ and $V^{-1}$. A similar argument can be used to prove the lower bound for the unique ground-state preparation of~$H$ from the lower bound for estimating the smallest eigenphase of~$U$.

\paragraph{Lower bound for ground-state energy estimation via Spectral Amplification.}

Improved quantum algorithms for ground-state energy estimation are known using 
sum-of-squares spectral  amplification (SOSSA)~\cite{king2026quantum}. Given a classical description of $H$, a preprocessing step (which may involve shifting so that $H \succeq 0$)
allows to express the Hamiltonian as a sum of squares. This construction essentially provides a ``square root'' operator $G$ that satisfies $H=G^\dagger G$.  
While the details of this preprocessing are not too important here, the key observation is that SOSSA subsequently treats $G$ as a black box and uses it to perform faster phase estimation. 
By leveraging the access to the square-root $G$, the eigenvalues of $H$ near zero can be ``amplified'',  reducing the precision requirements for phase estimation and lowering the query complexity of ground-state energy estimation.
These results raise a fundamental question: what is the minimum number of queries to the black box (or block-encoding) $G$  needed to estimate the ground-state energy of a Hamiltonian $H=G^\dagger G$, accounting for the approximation error, overlap, and error probability?

We can extend our results to this setting and answer this question as follows.  Consider the  instances where the Hamiltonian $H=G^\dagger G$ is such that
$G=\frac 1 {1+\sin (3\delta)}(\sin(3\delta)I - (U-U^{-1})/2\ri)$, where $U$ comes from one of our hard families and $\delta >0$. 
(Note that we do not have $U=e^{\ri H}$ here.)
The division by $1+\sin (3\delta)$ ensures $\norm{G}\leq 1$, which is necessary for a block-encoding.
The block-encoding of $G$ can be easily obtained with one use of $U$ and $U^{-1}$ like before, again assuming black-box access to these unitaries. 
Our lower bound 
$\Omega(\log(1/\eps)/\gamma \delta)$ on queries to the unitaries $U$ and $U^{-1}$ comes from
the decision problem
of determining whether $U$ has an eigenphase $3 \delta$ or $U=I$.
 This problem is equivalent to determining whether the smallest eigenvalue of the current $H$ is 0 or $\delta'=(\sin(3\delta)/(1+\sin(3\delta)))^2 \ge c \delta^2$, respectively, for some constant $c>0$ (assuming e.g.\ $3 \delta \le \pi/2$).  
Consequently, $\Omega(\log(1/\eps)/\gamma \sqrt{\delta'})$ queries to the black box or block-encoding of $G$ are necessary 
for estimating the ground-state energy of $H$ up to $\pm \delta'/3$. Given that SOSSA provides the same quadratic improvement in error dependence, this lower bound demonstrates that SOSSA is near-optimal  when treating the operator $G$ derived from $H$ as a black box. 

\paragraph{Follow-up work about ground-state preparation.}
Very recently, we teamed up with the authors of~\cite{JW:optQPE} to also obtain optimal algorithms for preparing ground states~\cite{JSWWstateprep}. In fact we give two different algorithms for this task, one based on a robust version of amplitude amplification~\cite{hmw:berrorsearch} and one based on transducers.
Concurrently, essentially the same result appeared as~\cite{chen2026optimal}

\section{Preliminaries}\label{sec:prelim}

A real univariate polynomial of degree $d$ is a function $p:\mathbb{R}\to\mathbb{R}$ of the form
\[
p(z)=\sum_{i=0}^d a_i z^i,
\]
for some real coefficients $a_0,\ldots,a_d$. 

A complex $N$-variate multilinear polynomial of degree $d$ is a function $Q:\mathbb{R}^N\to\mathbb{C}$ of the form
\[
Q(x)=\sum_{S\subseteq[N], |S|\leq d} c_S \prod_{j\in S}x_j,
\]
for some complex coefficients $c_S$. The term $\prod_{j\in S}x_j$ is called a monomial, with the convention that it is the constant-1 function  for $S=\emptyset$.
Such an $N$-variate polynomial~$Q$ can be ``symmetrized'' to a univariate polynomial $q$ such that $q(|x|)=Q(x)$ 
for all $x\in\01^N$, where $|x|$ denotes the Hamming weight (number of 1s) of $x$, as follows.
Observe that the $N$-variate polynomial
\[
\frac{1}{N!}\sum_{\pi\in S_N} Q(\pi(x))
\]
has a coefficient for the $S$-monomial that only depends on the size $|S|$. Hence this polynomial actually is a linear combination of $\binom{|x|}{|S|}$, which counts the number of degree-$|S|$ monomials that take value~1 on an input of weight~$|x|$, and which is a univariate polynomial $|x|(|x|-1)\cdots(|x|-|S|+1)/|S|!$ in $|x|$ of degree $|S|$.

Coppersmith and Rivlin~\cite{coppersmith&rivlin:poly} proved the following inequality, which allows to turn an upper bound for $p$ on an equally-spaced set of points into a (larger) upper bound on the whole real interval:
\begin{quote}
There exist universal constants $a,b>0$ such that every polynomial $p$ of degree $d \leq n$ that has
absolute value
$|p(i)| \leq 1$ for all integers $i \in\{0,\ldots,n\}$,\\ 
satisfies
$|p(z)| \leq ae^{bd^2/n}$ for all real $z \in [0, n]$.
\end{quote}
The degree-$d$ Chebyshev polynomial of the first kind can be defined as
\[
T_d(z) = \frac{1}{2}\left( (z+\sqrt{z^2-1})^d + (z-\sqrt{z^2-1})^d\right).
\]
The minus sign in the second term lets the square-roots terms vanish if we work out the power-$d$ terms, so (despite appearances to the contrary) this is actually a polynomial of degree~$d$. 
Using $1+y\leq e^y$, it is easy to show
\[
T_d(1+\mu)\leq e^{2d\sqrt{2\mu+\mu^2}}.
\]
These Chebyshev polynomials have absolute value $\leq 1$ on the domain $z\in[-1,1]$ (which is easier to see from their equivalent definition $T_d(z)=\cos(d\arccos(z))$).
They have the extremal property that no degree-$d$ polynomial $p$ with that same property can grow faster than $T_d$: $|p(1+\mu)|\leq T_d(1+\mu)$ for all $\mu\geq 0$.

\section{Proof for the case with non-unique ground state, $N=\log(1/\eps)/\gamma^2$}\label{sec:lowerboundnonunique}

In this section we prove the following.

\begin{theorem}
There exists a family of $O(\log(1/\eps)/\gamma^2)$-dimensional diagonal Hamiltonians~$H$, of operator norm $\norm{H}=O(1)$, such that every quantum algorithm that finds the ground-state energy of $H$ within $\pm\delta$ with success probability $\geq 1-\eps$, using $T$ controlled  applications of $U=e^{\ri H}$ and $U^{-1}$ and unlimited applications of a unitary $A$ and $A^{-1}$, for which $A\ket{0}$ is promised to have overlap at least $\gamma$ with the ground space of $H$,
needs at least 
\[
T=\Omega(\log(1/\eps)/\gamma\delta).
\]
\end{theorem}

For simplicity we will prove the lower bound for approximating the top eigenvalue of~$H$ rather than the bottom eigenvalue (i.e., ground-state energy) since it will be convenient to use 0 as a baseline, but those problems are easily seen to be equivalent.
At various points in the proofs we will assume $\delta,\gamma,\eps$ are at most some small constant; this is justified because the lower bound was already known if any one of these 3 parameters is $\Omega(1)$.
Choose $N=\log(1/\eps)/\gamma^2$ (ignore rounding for simplicity). Consider the  $(N+1)$-dimensional Hilbert space spanned by basis states $\ket{0},\ket{1},\ldots,\ket{N}$.
Our hard family of instances is the following.
For $x=x_1\ldots x_N\in\01^N$ and $\theta\in[0,2\pi)$, define
\[
\mbox{Hamiltonian }H=\mathrm{diag}(0,\{\theta x_j\}_{j\in[N]})
\mbox{ and corresponding  unitary }U=e^{\ri  H}=\mathrm{diag}(1,\{e^{\ri  \theta x_j}\}_{j\in[N]}).
\]
If $x\neq 0^N$ then the top eigenspace of~$H$ is $T_x=\mathrm{span}\{\ket{j}:x_j=1\}$, which has dimension equal to the Hamming weight~$|x|$.
Note that $U$ corresponds to a $\theta/\pi$-fractional query to the string $x\in\01^N$.
Because $U$ always acts like identity on the basis state $\ket{0}$, we do not need to introduce an additional controlled-$U$ with a separate control qubit. Accordingly, 
our lower bound on the number of applications of $U$ and $U^{-1}$ will be a lower bound on the number of applications of their controlled versions as well. 

Our lower bound will be for distinguishing $x=0^N$ (where $U=I$) from $x$s of Hamming weight $|x|\geq \gamma^2 N$ for the specific value $\theta=3\delta$.
Define $A\ket{0}=\frac{1}
{\sqrt{N}}\sum_{j=1}^N\ket{j}$ to be the guiding state.
This has squared overlap $\sum_{j:x_j=1} 1/N\geq\gamma^2$ with the top eigenspace, as promised. However, the guiding state $A\ket{0}$ is actually independent of $U$ and easy to generate for any algorithm, so we can just ignore it for the remainder of this lower bound proof.

Now consider an algorithm that makes $T$ calls to $U$ and $U^{-1}$, thinking of $x$ and $\theta$ as variables, and that (whenever $U$ comes from the above hard family) approximates the maximal eigenvalue of~$H$ to within $\pm\delta$ with success probability $\geq 1-\eps$.
After this algorithm we output 1 if the estimate is $\leq 1.5\delta$ and we output 0 if the estimate is $> 1.5\delta$.
This algorithm distinguishes $x=0^N$ from $x$s of weight $\geq\gamma^2 N$ when $\theta=3\delta$.
By pushing intermediate measurements to the end, we may assume it is fully coherent, with only a measurement of the first qubit in the final state to determine the binary output.

As usual when applying the polynomial method~\cite{BBCMW01}, we start by showing that the acceptance probability of a few-query algorithm corresponds to a low-degree polynomial. The following variant of the method handles fractional queries. 

\begin{lemma}\label{lem:Tfracto}
Consider a quantum algorithm that makes $T$ $\phi$-fractional queries to $x\in\01^N$, for some fixed $\phi\in[0,1]$. 
Let $P(x)$ be the probability that it outputs 1 on input $x$.
For every $\eps\in(0,1)$, there exists a multilinear real polynomial $Q(x)$ of degree $D=O(T\phi+\log(1/\eps))$ such that $Q(x)\geq 0$ and $|P(x)-Q(x)|\leq\eps$ for all $x\in\01^N$.
\end{lemma}

\begin{proof}
Let $O_x=\mbox{diag}(1,\{(-1)^{x_j}\}_{j\in[N]})$ be the usual (controlled) phase query to $x\in\01^N$, and 
\begin{equation}\label{eq:decompOxdelta}
O_x^\phi=\mbox{diag}(1,\{e^{\ri \pi\phi x_j}\}_{j\in[N]})=\frac{1+e^{\ri \pi\phi}}{2}I+\frac{1-e^{\ri \pi\phi}}{2}O_x
\end{equation}
be the fractional query.
We have a similar decomposition for the inverse $O_x^{-\phi}$.
Writing $O_x^{\pm\phi}$ as  placeholders for something that could be $O_x^{\phi}$ or $O_x^{-\phi}$, the $T$-query algorithm corresponds to a unitary 
\[
{\cal A}=W_TO_x^{\pm\phi} W_{T-1}O_x^{\pm\phi}\cdots O_x^{\pm\phi} W_1 O_x^{\pm\phi} W_0,
\]
for some fixed, input-independent unitaries $W_T,\ldots,W_0$, applied to fixed initial state $\ket{0}$ and followed by a measurement of the first qubit to produce the binary output.
Use the decomposition of $O_x^{-\phi}$ into $I$ and $O_x$ of Eq.~\eqref{eq:decompOxdelta} and its analogue for the inverse, decomposing $\cal A$ as a sum of $2^T$ terms, and let ${\cal A}_K$ be the truncated sum obtained by omitting all terms that have more than $K=cT\phi+\log(3/\eps)$ $O_x$'s in them, for some sufficiently large constant $c$ to be determined later (we will choose $c$ such that this $K$ is an integer).
Using triangle inequality and the fact that the $W_j$'s are unitary, and $\binom{T}{k}\leq (eT/k)^k$, we show that ${\cal A}$ and ${\cal A}_K$ are close in operator norm:
\[
\norm{{\cal A}-{\cal A}_K}\leq \sum_{k=K+1}^T\binom{T}{k}\left|\frac{1+e^{\ri \pi\phi}}{2}\right|^{T-k}\cdot\left|\frac{1-e^{\ri \pi\phi}}{2}\right|^k
\leq \sum_{k=K+1}^T O(T\phi/k)^k\leq O( T\phi/K)^K\leq 2^{-K}\leq\eps/3,
\]
where the second and third inequality follow by choosing the constant $c$ in the definition of $K$ large enough so that the $O( T\phi/K)$ term is $\leq 1/2$.

By the usual polynomial-method inductive argument, the final entries of the vector ${\cal A}_K\ket{0}$ can be shown to be complex $N$-variate multilinear polynomials in $x$ of degree $\leq K$: the initial amplitudes are constants independent of $\theta$ and $x$ (i.e., polynomials of degree~0), each $O_x$ increases the degree by at most~1, and the intermediate input-independent unitaries just take linear combinations of amplitudes so they do not increase degree.

Let $\Pi$ be the projector on the entries corresponding to basis states that start with a~1, and define $Q(x)=\norm{\Pi{\cal A}_K\ket{0}}^2$, which is the ``acceptance probability'' of ${\cal A}_K$ (we use scare quotes because ${\cal A}_K$ is not quite unitary). Being the sum-of-squares of degree-$K$ polynomials, this~$Q$ is a polynomial of degree $2K$. It can be multilinearized by changing all $x_j^d$ for $d>1$ to $x_j$ (which does not change the value on 0/1 input variables).
We have, for all $x\in\01^N$, 
\begin{align*}
|P(x)-Q(x)| & =\left|\norm{\Pi{\cal A}\ket{0}}^2 - \norm{\Pi{\cal A}_K\ket{0}}^2\right|\\
& =\left|\norm{\Pi{\cal A}\ket{0}} - \norm{\Pi{\cal A}_K\ket{0}}\right|\cdot \left(\norm{\Pi{\cal A}\ket{0}} + \norm{\Pi{\cal A}_K\ket{0}}\right)\\
& \leq \left|\norm{\Pi{\cal A}\ket{0} - \Pi{\cal A}_K\ket{0}}\right|\cdot    \left(\norm{\Pi{\cal A}\ket{0}} + \norm{\Pi{\cal A}_K\ket{0}}\right)\\
& \leq \left|\norm{{\cal A}\ket{0} - {\cal A}_K\ket{0}}\right|\cdot |1+1+\eps|\leq\eps,
\end{align*}
hence $Q$ $\eps$-approximates $P$.
\end{proof}

In order to lower bound $T$, we now need tools to lower bound the degree of the polynomial that comes out of the previous lemma.
To that end, define ${\cal K}=\{0,\gamma^2 N,\ldots,N\}$,
 and define $f:{\cal K}\to\01$ by $f(0)=1$ and $f(k)=0$ for all $k\in\{\gamma^2 N,\ldots,N\}$.
For a given $D$, let 
\[
\eps_D=\min_{\mbox{polynomial $p$ of degree}\leq D}\max_{k\in {\cal K}}|f(k)-p(k)|
\]
be the minimal worst-case error that a degree-$D$ real univariate polynomial can attain for approximating $f$.
Note that $\eps_D\leq 1/2$, since we can always take $p\equiv 1/2$ to be our approximating polynomial.
Using the tools from~\cite{BCWZ99} (see our Section~\ref{sec:prelim}) we can show the following lower bound on $\eps_D$.

\begin{lemma}\label{lem:CRChebyshev}
If $q$ is degree-$D$ real univariate polynomial such that $q(0)\geq 1$ and $|q(a)|\leq\eta$ for all $a\in[\gamma^2,1]$, then $\eta\geq 2^{-O(D\gamma)}$.
\end{lemma}

\begin{proof}
Rescaling the domain $[\gamma^2,1]$ to $[-1,1]$ in such a way that 0 gets sent to $1+\mu$ for $\mu=\Theta(\gamma^2)$, and dividing the value of the polynomial~$q$ by $\eta$, we obtain a degree-$D$ real polynomial~$r$ that takes values $r(z)\in[-1,1]$ on the interval $z\in[-1,1]$. By the extremal properties of Chebyshev polynomials we can bound:
\[
r(1+\mu)\leq T_{D}(1+\mu)\leq e^{2D\sqrt{2\mu+\mu^2}}=2^{O(D\gamma)}. 
\]
Hence we obtain
\[
1\leq q(0)=\eta\cdot r(1+\mu)\leq \eta \cdot 2^{O(D\gamma)},
\]
and rearranging gives the lemma.
\end{proof}

\begin{lemma}\label{lem:epsD}
If $D\leq (1-\gamma^2)N$, then $\eps_D\geq 2^{-O(D^2/N + D\gamma)}$.
\end{lemma}

\begin{proof}
Let $p$ be a degree-$D$ real polynomial achieving the minimal error $\eps_D$ in approximating~$f$.
Since $|p(z)|\leq \eps_D$ for all $(1-\gamma^2)N+1$ \emph{integers} $z\in\{\gamma^2N,\ldots,N\}$, by the Coppersmith-Rivlin inequality we have $|p(z)|\leq \eps_D \cdot 2^{O(D^2/N)}$ for all \emph{real} $z\in[\gamma^2N,N]$. Rescaling the domain $[\gamma^2N,N]$ to $[\gamma^2,1]$ and dividing by $1-\eps_D\geq 1/2$, we obtain a degree-$D$ real polynomial~$p$ that takes values $p(0)\geq 1$ and $|p(z)|\leq\eps_D \cdot 2^{O(D^2/N)}:=\eta$ for all $z\in[\gamma^2,1]$. Applying Lemma~\ref{lem:CRChebyshev} gives
\[
\eta\geq 2^{-O(D\gamma)},
\]
and rearranging gives the lemma.
\end{proof}

Now consider an algorithm ${\cal A}$ with $T$ $\phi$-fractional queries for our hard family, with $\phi=\theta/\pi$.
Its acceptance probability $P(x)$ for fixed $\theta=3\delta$ $\eps$-approximates~$f(|x|)$: on $x=0^N$ the algorithm returns~1 with probability $\geq 1-\eps$, and if $x$ has Hamming weight $|x|\geq\gamma^2 N$ then  the algorithm returns~1 with probability $\leq \eps$. The polynomial $Q(x)$ of degree $D=O(T\delta+\log(1/\eps))$
that we get from Lemma~\ref{lem:Tfracto} $\eps$-approximates $P(x)$. 
As explained in  Section~\ref{sec:prelim}, we can symmetrize $Q(x)$ to a univariate polynomial~$q$ of degree $\leq D$ such that $q(|x|)=Q(x)$, and hence $q$ is nonnegative and it $2\eps$-approximates~$f$. By Lemma~\ref{lem:epsD}, the approximation error of $q$ (w.r.t.\ $f$) is $\geq 2^{-O(D^2/N + D\gamma)}$, hence 
\[
\log(1/2\eps)\leq O(D^2/N + D\gamma)\leq O(T^2\delta^2\gamma^2/\log(1/\eps) + \log(1/\eps)\gamma^2 + T\delta\gamma + \log(1/\eps)\gamma),
\]
where we used  $N=\log(1/\eps)/\gamma^2$. This implies $T=\Omega(\log(1/\eps)/\gamma\delta)$ if $\eps$ and $\gamma$ are below a sufficiently small constant (the latter ensures that the $O(\log(1/\eps)\gamma)$ term on the right-hand side is  $\leq \log(1/2\eps)/2$; if $\eps=\Omega(1)$ or $\gamma=\Omega(1)$ then the lower bound was already known, as mentioned in the introduction).

\section{Proof for the case with unique ground state, $N=\log(1/\eps)^2/\gamma^2$}\label{sec:lowerboundunique}

In this section we show the same lower bound $T=\Omega(\log(1/\eps)/\gamma\delta)$ for the special case where the $H$ have a unique ground state, whose eigenvalue is bounded away from all other eigenvalues.
The price to pay is that we have to make the dimension larger by a $\log(1/\eps)$-factor now:

\begin{theorem}
There exists a family of $O(\log(1/\eps)^2/\gamma^2)$-dimensional Hamiltonians~$H$, of operator norm $\norm{H}=O(1)$, that have a unique ground state and a ground-state energy that is at least $\delta$ less than the second eigenvalue, such that every quantum algorithm that finds the ground-state energy of $H$ within $\pm\delta$ with success probability $\geq 1-\eps$, using $T$ controlled  applications of $U=e^{\ri H}$ and $U^{-1}$ and unlimited applications of a unitary $A$ and $A^{-1}$ promised that $A\ket{0}$ has overlap at least $\gamma$ with the ground space of $H$,
needs at least 
\[
T=\Omega(\log(1/\eps)/\gamma\delta).
\]
\end{theorem}

Again we lower bound the equivalent problem of approximating $\lambda_{\max}$ of~$H$ rather than $\lambda_{\min}$.
Consider Hilbert space spanned by $\ket{0},\ket{1},\ldots,\ket{N}$. The unitaries $U$ in the hard family will now be 
\begin{equation}\label{eq:decompUunique}
U=I+(e^{\ri \theta}-1)\ketbra{\psi_{a\omega j}}{\psi_{a\omega j}},
\end{equation}
where 
\[
\ket{\psi_{a\omega j}}=\sqrt{a}\ket{0}+\omega\sqrt{1-a}\ket{j}, 
\]
for some $a\in[0,1]$, $j\in[N]$, $\omega\in\{-1,1\}$.
For $\theta>0$, $U$ has one eigenstate $\ket{\psi_{a\omega j}}$ with eigenvalue $e^{\ri \theta}$, and on the rest $U$ acts like identity.
This~$U$ is no longer diagonal in the computational basis, though it has at most 2 nonzero off-diagonal entries. 
For $\theta=0$, we have $U=I$, which of course does not have a unique ground state, but our proof treats it as a degenerate baseline to show the hardness of approximating a unique ground-state value $\theta>0$.

Let $A\ket{0}=\ket{0}$.
This guiding state is useless, but as long as $a\geq \gamma^2$ it does have the required overlap $\geq \gamma$ with the unique top eigenvector $\ket{\psi_{a\omega j}}$.

Now fix an algorithm with $T$ applications of $U$ and $U^{-1}$, with output 1 corresponding to an estimate close to~0 as in the previous section. Let $p(a,\theta)$ be the acceptance probability of the algorithm, averaged uniformly over all $j\in[N]$ and $\omega\in\{-1,1\}$.
From the correctness of the algorithm up to error probability $\eps$, 
we know that $p(a,0)\in[1-\eps,1]$ for all $a\in[0,1]$ (incl.\ for $a=0$), and $p(a,3\delta)\in[0,\eps]$ for all $a\in[\gamma^2,1]$.
The correctness of the algorithm does not tell us whether $p(0,3\delta)$ is large or small, so we will make the following case distinction.

\medskip

{\bf Case 1: $p(0,3\delta)\in[0,1/2)$.}
Note that for $a=0$, we have
$U=I+(e^{\ri\theta}-1)\ketbra{j}{j}$. If $\theta=0$, then $U=I$, which corresponds (rather trivially) to making $3\delta/\pi$-fractional phase queries to the string $x=0^N$. If $\theta=3\delta$, then $U$ corresponds to making $3\delta/\pi$-fractional phase queries to the string $x=e_j$, the $N$-bit string that has a 1 only at the $j$th position. The fact that $p(0,3\delta)< 1/2$ while $p(0,0)\geq 1-\eps$ means that the algorithm distinguishes these two cases with small constant  error probability, for which a lower bound of $T=\Omega(\frac{1}{\delta}\sqrt{N})$ fractional queries is well-known (see, e.g., \cite[Appendix~B]{lmrss:stateconv} combined with the $\Omega(\sqrt{N})$ bound for search, or \cite[Lemma 3.5]{Mande2026tightboundsquantum}).

\medskip

{\bf Case 2: $p(0,3\delta)\in[1/2,1]$.}
For this case we will prove $T=\Omega(\log(1/\eps)/\gamma\delta)$, which is substantially more involved than Case~1.

The main thing is the following modified translation of the acceptance probability of a $T$-query algorithm to a polynomial, which simultaneously allows us to focus on the case of constant~$\theta$. This is the analogue of Lemma~\ref{lem:Tfracto} from the previous section.

\begin{lemma}
Consider a quantum algorithm that makes $T$ applications of $U$ or $U^{-1}$, for some fixed $\theta$ and variable $a,j,\omega$. 
Let $P(a)$ be the probability, averaged uniformly over $j$ and $\omega$, that it outputs 1.
For every $\eps\in(0,1)$, there is a real univariate polynomial $Q(a)$ of degree $D=O(T\theta+\log(1/\eps))$ such that $Q(a)\geq 0$ and $|P(a)-Q(a)|\leq\eps$ for all $a\in[0,1]$.
\end{lemma}

\begin{proof}
Let $R_{a\omega j}=I-2\ketbra{\psi_{a\omega j}}{\psi_{a\omega j}}$ be the reflection about the state $\ket{\psi_{a\omega j}}$.
Note that $U$ can be decomposed as
\begin{equation}\label{eq:decompUtheta}
U=\frac{1+e^{\ri \theta}}{2}I+\frac{1-e^{\ri \theta}}{2}R_{a\omega j},
\end{equation}
and a similar decomposition of $U^{-1}$ with a minus in the exponent.
Writing $U^{\pm}$ as placeholders for something that could be $U$ or $U^{-1}$, the $T$-query algorithm corresponds to a unitary 
\[
{\cal A}=W_T U^{\pm} W_{T-1}\cdots U^{\pm} W_1 U^{\pm} W_0,
\]
for some fixed, input-independent unitaries $W_T,\ldots,W_0$, applied to fixed initial state $\ket{0}$ and followed by a measurement of the first qubit to produce the binary output.
Use the decomposition of $U$ and $U^{-1}$, decomposing $\cal A$ as a sum of $2^T$ terms, and let ${\cal A}_K$ be the truncated sum obtained by omitting all terms that have more than $K=cT\theta+\log(3/\eps)$ $R_{a\omega j}$'s in them, for some sufficiently large constant $c$ to be determined later.
Using triangle inequality and the fact that the $W_j$'s are unitary, and $\binom{T}{k}\leq (eT/k)^k$, we show that ${\cal A}$ and ${\cal A}_K$ are close in operator norm:
\[
\norm{{\cal A}-{\cal A}_K}\leq \sum_{k=K+1}^T\binom{T}{k}\left|\frac{1+e^{\ri \theta}}{2}\right|^{T-k}\cdot\left|\frac{1-e^{\ri \theta}}{2}\right|^k
\leq \sum_{k=K+1}^T O(T\theta/k)^k\leq O( T\theta/K)^K\leq 2^{-K}\leq\eps/3,
\]
where the second and third inequality follow by choosing the constant $c$ in the definition of $K$ large enough so that the $O( T\theta/K)$ term is $\leq 1/2$.

By a variant of the usual polynomial-method inductive argument~\cite{BBCMW01}, the final entries of the vector ${\cal A}_K\ket{0}$ can be written as
\[
\sum_{d=0}^{K}\sum_{d_1,d_2,d_3\geq 0:d_1+d_2+d_3=d} b_{d_1,d_2,d_3}a^{d_1}(1-a)^{d_2}(\omega\sqrt{a(1-a)})^{d_3},
\]
for some coefficients $b_{d_1,d_2,d_3}\in\mathbb{C}$. The acceptance probability is the sum of squared moduli of such amplitudes, which can then be written as
\[
\sum_{d=0}^{2K}\sum_{d_1,d_2,d_3\geq 0:d_1+d_2+d_3=d} c_{d_1,d_2,d_3}a^{d_1}(1-a)^{d_2}(\omega\sqrt{a(1-a)})^{d_3},
\]
for some $c_{d_1,d_2,d_3}\in\mathbb{R}$. For the terms where $d_3$ is even, $(\omega\sqrt{a(1-a)})^{d_3}=(a(1-a))^{d_3/2}$ is a polynomial in $a$ of degree $d_3$, hence that whole term is a polynomial in $a$ of degree $d_1+d_2+d_3=d$. The terms where $d_3$ is odd will vanish when we average the acceptance probability over $\omega\in\{-1,1\}$, since $\sum_{\omega\in\{-1,1\}}(\omega\sqrt{a(1-a)})^{d_3}=\sum_{\omega\in\{-1,1\}}\omega(\sqrt{a(1-a)})^{d_3}=0$. This gives a polynomial $Q_j(a)$ that depends on $j$. Averaging over $j$ gives the polynomial $Q(a)$ of degree $\leq 2K=O(T\theta+\log(1/\eps))$. It $\eps$-approximates $P(a)$ in the same way as at the end of the proof of Lemma~\ref{lem:Tfracto}.
\end{proof}

To obtain our lower bound for Case~2, we will lower bound the degree $D=O(T\theta+\log(1/\eps))$ of the polynomial $Q$ that comes out of the previous lemma. This $Q$ $\eps$-approximates $P$. Because $P(0)\in[1/2,1]$ and $P(a)\leq \eps$ for all $a\in[\gamma^2,1]$, the degree-$D$ polynomial $q(a)=Q(a)/Q(0)$ satisfies $q(0)=1$ and $|q(a)|\leq (P(a)+\eps)/(P(0)-\eps)\leq 8\eps$ for all $a\in[\gamma^2,1]$ (assuming $\eps\leq 1/4$). By Lemma~\ref{lem:CRChebyshev}, the approximation error of $q$ is $\geq 2^{-O(D\gamma)}$. Hence 
\[
\log(1/8\eps)\leq O(D\gamma)\leq O((T\delta + \log(1/\eps))\gamma)
\]
which implies $T=\Omega(\log(1/\eps)/\gamma\delta)$ if $\eps$ and $\gamma$ are below a sufficiently small constant.

\bigskip

{\bf Combining the two cases.}
Case~1 gives $T=\Omega(\frac{1}{\delta}\sqrt{N})$ and Case~2 gives $T=\Omega(\log(1/\eps)/\gamma\delta)$, so the overall lower bound is the minimum of those two disjoint cases. Choosing $N=\log(1/\eps)^2/\gamma^2$ makes these two lower bounds equal, giving the claimed overall bound
\[
T=\Omega(\log(1/\eps)/\gamma\delta).
\]
Note that the $\eps$-dependence appears only through Case~2, and the dimension only appears through Case~1. Choosing smaller dimension $N$ results in a worse lower bound because Case~1 then determines the minimum, and larger $N$ does not result in a better bound because Case~2 then determines the minimum.

\section{Discussion and future work}

We provided a lower bound to the quantum query complexity of ground-state energy estimation of a Hamiltonian $H$
under the assumption that we can generate a guiding state
that has overlap at least $\gamma >0$ with the ground state or ground space. 
Our lower bound $\Omega(\log(1/\eps)/\gamma \delta)$ includes all relevant asymptotic
parameters, namely $\gamma$, the estimation error or spectral gap $\delta$, and the error probability $\eps$.  It assumes controlled access to unitaries like $U=e^{\pm \ri H}$ or a block-encoding of $H$. 
It matches the very recent upper bound of Jeffery and Witteveen~\cite{JW:optQPE} for guided ground-state energy estimation up to constant factors, and hence is optimal in all three parameters. 
The lower bound also extends to the task of preparing the ground \emph{state} of~$H$ up to  trace distance~$\eps$.

We note that
another relevant access model in quantum simulation is the sparse-matrix access model, where the Hamiltonian~$H$ is $d$-sparse (meaning each row and column has at most $d$ nonzero entries) and one can ask queries of the form ``what is the $k$th nonzero entry in the $j$th row of $H$''. 
Despite the fact that the $H$ in our hard families are 1- or 2-sparse, our lower bound does not apply to that access model and does not include the overhead due to $d$;
 proving such a lower bound is left as an open problem.

Somewhat surprisingly, and maybe disappointingly, the hard families for both of our lower bounds involve guiding states that are effectively independent of the ground state (and of its energy), despite having the required $\geq\gamma$ overlap with the ground state or ground space. This suggests that the standard requirement in the guided Hamiltonian problem of large overlap may not be the best way to set up the problem:
somehow the guiding state should give genuinely useful information about the ground state, not just overlap with uninformative parts of it (as in the hard family of Section~\ref{sec:lowerboundunique}, where the overlap is always in the uninformative state $\ket{0}$).
We leave the question what could be a better alternative to a simple overlap requirement for future work.

\bigskip
\bigskip

\noindent{\bf Acknowledgment}: 
We thank Robin Kothari, Guang Hao Low, and Gabriel Waite for helpful comments, and Stacey Jeffery and Freek Witteveen for sending us a version of~\cite{JW:optQPE}, for helpful comments, and for coordinating arXiv submission.

\medskip

\noindent{\bf AI statement}: ChatGPT 5.6 Pro generated both hard families and very convoluted and long (20+ pages) but largely correct proofs of the lower bounds, building on top of existing lower-bound techniques from~\cite{BCWZ99,Mande2026tightboundsquantum}, dual polynomials, and trigonometric polynomials. 
It was first prompted to prove the conjectured lower bound by combining two specific lemmas from~\cite{Mande2026tightboundsquantum} and using tools from trigonometric polynomials and Chebyshev polynomials.
The main new idea that ChatGPT came up with is the observation that one can write the acceptance probability as a hybrid trigonometric-real polynomial in $\theta$ and $x$ with a zero of large multiplicity at $\theta=0$; this zero can then be used together with Taylor's theorem and Bernstein's inequality to upper bound the acceptance probability at $\theta=0$.
However, we eventually managed to whittle down the proofs to something much simpler that doesn't need this new idea anymore.
The paper writing is fully human.
ChatGPT was also used to critique intermediate drafts, revealing some small errors in the writing that we corrected.

\bibliographystyle{alpha}

\bibliography{bibo}

\end{document}